\documentclass[11pt]{article}
\usepackage{amssymb}

\newcommand{\R}{\Bbb R}

\newcommand{\mc}{\mathcal}

\newcommand{\be}{\begin{equation}}
\newcommand{\en}{\end{equation}}
\newcommand{\D}{{\mc D}}

\newcommand{\I}{{\cal I}}

\newcommand{\Lc}{{\cal L}}

\newcommand{\N}{\mathbb N}
\newcommand{\Hil}{\mc H}
\newcommand{\C}{{\cal C}}

\newcommand{\bea}{\begin{eqnarray}}
\newcommand{\ena}{\end{eqnarray}}
\newtheorem{thm}{Theorem}

\newtheorem{lemma}[thm]{Lemma}
\newtheorem{prop}[thm]{Proposition}

\catcode `\@=11 \@addtoreset{equation}{section}

\catcode `\@=12
\newenvironment{proof}{\noindent {\bf Proof --}}{\hfill$\square$ \vspace{3mm}\endtrivlist}

\begin{document}

\begin{center}
{\Large \bf A bounded version of bosonic creation and annihilation operators and their related
{\em quasi-coherent states}}   \vspace{2cm}\\

{\large F. Bagarello}
\vspace{3mm}\\
  Dipartimento di Metodi e Modelli Matematici,
Facolt\`a di Ingegneria,\\ Universit\`a di Palermo, I - 90128  Palermo, Italy\\
E-mail: bagarell@unipa.it\\home page:
www.unipa.it$\backslash$\~\,bagarell
\vspace{2mm}\\
\end{center}

\vspace*{2cm}

\begin{abstract}
\noindent Coherent states are usually defined as eigenstates of an
unbounded operator, the so-called annihilation operator. We
propose here possible constructions of {\em quasi-coherent
states}, which turn out to be {\em quasi} eigenstate of a
\underline{bounded} operator related to an annihilation-like
operator. We use this bounded operator to construct a sort of
modified harmonic oscillator and we  analyze the dynamics of this
oscillator from an algebraic point of view.

\end{abstract}

\vfill

\newpage

\section{Introduction}

One of the very few undergraduate examples of quantum mechanical
systems for which a complete solution is known, the harmonic
oscillator, suggests the introduction of the so-called {\em ladder
operators}, the annihilation and creation operators $a$ and
$a^\dagger$. They turn out to obey the commutation rule
$[a,a^\dagger]=\I$ which implies that they are unbounded operators
so that, for instance, they are not everywhere defined. Exactly
the same kind of operators appear also in the analysis of a rather
different system, a gas of bosons, and they are also almost
everywhere in quantum optics. This explain the amount of papers
related to many aspects  of these operators and to other
quantities which are related to $a$ and $a^\dagger$, like, for
instance, the coherent states, in any of their forms, and the
squeezed states.

Generalized creation and annihilation operators, $A$ and
$A^\dagger$, have been  introduced in recent years and used to
construct generalized version of coherent states, see for instance
\cite{ali} and references therein. However, also these new
operators are usually unbounded.

In this paper we consider all these objects from the point of view
of bounded operators. In other words, we propose a natural cutoff,
arising from $A$ and $A^\dagger$ themselves, and we use this
cutoff to build up new bounded operators, their free quantum
evolution and what we will call the {\em quasi-coherent states}.

This paper is organized as follows: in the next section we discuss
the mathematical background and we introduce the cutoff. In
Section III  we analyze the algebraic dynamics related to the free
regularized hamiltonian using topological results connected to the
so called quasi-* algebras, \cite{bagJMP,bagtra}, both at a finite
level and at the level of the derivations. Section IV is devoted
to the introduction and the analysis of {\em quasi-coherent
states},  using both the Ali's and the Gazeau-Klauder's languages,
\cite{ali,GK}. We conclude the paper with a short section where
some connections with quons are discussed,  \cite{moh,fivetc}, and
where it is shown that quons are just particular cases of a much
more general situation.

\section{Mathematical ingredients}

Let $\Hil$ be a separable Hilbert space and $\{\Phi_n,\,n\in
\N_0,\, \N_0=0,1,2\ldots\}$ an o.n. basis of $\Hil$. We define the
following operators \be P_{i,j}=|\Phi_i><\Phi_j|,\quad\quad
P_i=P_{i,i}, \quad\quad Q_L=\sum_{i=0}^L\,P_i,\label{21}\en where
we have used the Dirac {\em bra-ket} notation. They satisfy the
following properties: \be \|P_{i,j}\|=1,\quad\|Q_L\|=1, \quad
P_{i,j}P_{k,l}=\delta_{j,k}\,P_{i,l},\quad  P_i^\dagger=P_i, \quad
Q_L^\dagger=Q_L\quad \forall i,j,k,l,L\in\N_0,\label{22}\en as
well as $P_i^2=P_i$ and $Q_L^2=Q_L$. Therefore both $P_i$ and
$Q_L$ are orthogonal projectors. It is clear that $P_i$ projects
on the subspace of $\Hil$ generated by the single vector $\Phi_i$,
while $Q_L$ projects on the subspace generated by
$\{\Phi_0,\Phi_1,\ldots,\Phi_L\}$.

Let now $\{x_l,\,l\in\N_0\}$ be a generic sequence of non negative
numbers, $x_l\geq 0$ for all $l\in\N_0$. We will assume that only
$x_0$ could be equal to zero, while all the other coefficients are
strictly positive. We put \be
A_L=\sum_{l=0}^L\,\sqrt{x_{l+1}}\,P_{l,l+1},\quad \mbox{ so that
}\quad
A_L^\dagger=\sum_{l=0}^L\,\sqrt{x_{l+1}}\,P_{l+1,l}\label{23}\en
Both these operators are bounded, as far as $L<\infty$. Indeed a
very easy estimate is the following: $\|A_L\|=\|A_L^\dagger\|\leq
\sum_{l=0}^L\,\sqrt{x_l} $, which is a consequence of the fact
that $\|P_{i,j}\|=1$ for all $i$ and $j$. However, we can do
better since, for any $\Phi\in\Hil$,
$$
\|A_L\phi\|\leq
\sum_{l=0}^L\,\sqrt{x_{l+1}}\|<\Phi_{l+1},\phi>\,\Phi_{l+1}\|=\sum_{l=0}^L\,\sqrt{x_{l+1}}\,
|<\Phi_{l+1},\phi>|$$ $$
\leq\sqrt{\sum_{l=0}^L\,x_{l+1}}\,\,\sqrt{\sum_{l=0}^L\,|<\Phi_{l+1},\phi>|^{\,2}}\leq\|\phi\|\,
\sqrt{\sum_{l=0}^L\,x_{l+1}},
$$
which implies that, for all $L\in\N$, \be
\|A_L\|=\|A_L^\dagger\|\leq
d_L:=\sqrt{\sum_{l=0}^L\,x_{l+1}}.\label{24}\en We can also derive a
lower bound for $\|A_L\|$. For that it is enough to notice that
$\|A_L\|=\sup_{\phi\in\Hil,\,\|\phi\|=1}\,\|A_L\,\phi\|\geq
\|A_L\,\Phi_j\|$, for all values of $j\in\N_0$. Therefore
$\|A_L\|\geq \sqrt{ x_{j} }$, for all $j=0,1,\ldots,L$, which is
strictly positive if $x_0>0$. If the sequence $\{x_l,\,l\in\N_0\}$
belongs to $l^1(\N_0)$, then we can conclude that \be \sqrt{x_j}\leq
\|A_L\|\leq \sqrt{\|x\|_1},\label{25}\en for all possible
$j=0,1,\ldots,L$. We use $\|x\|_1$ to indicate the norm of
$\{x_l,\,l\in\N_0\}$ in $l^1(\N_0)$. If this is the case, then it is
trivial to check that $A_L$ converges in the uniform topology to a
bounded operator. Indeed, for all $L> M$, we find $\|A_L-A_M\|^2\leq
\sum_{l=M+1}^L\,x_{l+1}$, which goes to zero when
$L,M\rightarrow\infty$. Obviously the same conclusion does not hold,
in general, if $\{x_l,\,l\in\N_0\}\notin l^1(\N_0)$ and it {\em
surely} does not hold if  $\{x_l,\,l\in\N_0\}$ is not a bounded
sequence. In this case, in fact, because of the lower bound
$\|A_L\|\geq \sqrt{ x_{j} }$, $\forall j=0,1,\ldots,L$,
$\|A_L\|\rightarrow\infty$ when $L\rightarrow\infty$. This is rather
common in concrete applications since, for example for the {\em
standard} quantum harmonic oscillator, we just have $x_k=k$.
Therefore, an explicit choice of $\{x_l\}$ must be dictated by
whether we want $A_L$ to converge to a bounded operator or not
(which is what we really have in mind since, as we have just
discussed, it is what happens for an harmonic oscillator). A
possible way to define $A$ and $A^\dagger$ in this situation is
based on the fact that the o.n. basis $\{\Phi_n,\,n\in \N_0\}$
belongs to the domain of $A_L$ for each $L$. Introducing the
sequence $\chi_n$, which is equal to 1 for all $n\geq 0$ and 0
otherwise, we can easily check that \be A_L\,\Phi_k=\left\{
    \begin{array}{ll}
        0, \hspace{3.1cm} \mbox{ if } k=0,  \\
        \chi_{L+1-k}\,\sqrt{x_k}\,\Phi_{k-1}, \quad \mbox{ if } k>0,\\
       \end{array}
        \right. \quad\mbox{and}\quad
        A_L^\dagger\,\Phi_k=\chi_{L-k}\,\sqrt{x_{k+1}}\,\Phi_{k+1}.\label{26}
\en As it is clear, they behave like creation and annihilation
operators but with two significant differences: they give a non
zero result only for a finite set of $k$'s, corresponding to a
finite dimensional subset of $\Hil$. The second difference is that
the natural numbers $k$ are here replaced with a more general
sequence of non negative numbers, $k\rightarrow x_k$, as in
\cite{ali} and references therein.

The above formulas admit limit when $L\rightarrow\infty$, since
$\lim_{L\rightarrow\infty}\,\chi_{L-k}=\lim_{L\rightarrow\infty}\,\chi_{L+1-k}=1$
for all fixed $k$. This is the way in which we introduce here the
operators $A$ and $A^\dagger$ if $\{x_l,\,l\in\N_0\}\notin
l^1(\N_0)$: we just put \be A\,\Phi_k=\left\{
    \begin{array}{ll}
        0,\hspace{1.8cm}\mbox{ if } k=0,  \\
        \sqrt{x_k}\,\Phi_{k-1}, \quad \mbox{ if } k>0\\
       \end{array}
        \right. \quad\mbox{and}\quad
        A^\dagger\,\Phi_k=\sqrt{x_{k+1}}\,\Phi_{k+1}.\label{27}
\en Of course, this does not imply  that the definition (\ref{23})
extends to $A=\sum_{l=0}^\infty\,\sqrt{x_{l+1}}\,P_{l,l+1}$,
because the series does not converge, for generic $x_j$, in the
usual topologies of bounded operators.

It is straightforward to check that the operators $A$ and $A_L$
are related in the following way: \be A_L=Q_{L+1}\,A\,Q_{L+1},
\hspace{1cm} A_L^\dagger=Q_{L+1}\,A^\dagger\,Q_{L+1}.\label{28}
\en The following equalities moreover hold true: \be
A\,Q_{L+1}=Q_L\,A,\quad Q_{L+1}\,A^\dagger=A^\dagger\,Q_L,\quad
A\,P_{l}=P_{l-1}\,A,\quad P_l\,A^\dagger=A^\dagger\,P_{l-1}
\label{29} \en whose proof, again, is left to the reader. Notice
that the first equality implies in particular that \be
A^\dagger_L\,A_L= Q_{L+1}\,A^\dagger\,A\, Q_{L+1}, \quad
A_L\,A_L^\dagger= Q_{L}\,A\,A^\dagger\, Q_{L+1}.\label{210}\en
Indeed we have $A^\dagger_L\,A_L= Q_{L+1}\,A^\dagger\,Q_{L+1}\,
Q_{L+1}\,A\,Q_{L+1} =Q_{L+1}\, A^\dagger\,Q_{L+1}\,A\, Q_{L+1}=
Q_{L+1} \,A^\dagger\,A\, Q_{L+2}\\ \, Q_{L+1}=
Q_{L+1}\,A^\dagger\,A\, Q_{L+1}$, while
$A_L\,A_L^\dagger=Q_{L+1}\,A\,Q_{L+1}\,Q_{L+1}\,A^\dagger\,Q_{L+1}=
Q_{L+1}\,A\,Q_{L+1}\,A^\dagger\,Q_{L+1}\\=
Q_{L+1}\,Q_{L}\,A\,A^\dagger\, Q_{L+1}= Q_{L}\,A\,A^\dagger\,
Q_{L+1}$. Other relevant formulas are the following commutation
rules

\be        [Q_{L},A]=P_L\,A=A\,P_{L+1},\quad
[Q_{L},A^\dagger]=-P_{L+1}\,A^\dagger=-A^\dagger\,P_L,
\label{211}\en and \be   [Q_L,A^\dagger\, A]=[Q_L,
A\,A^\dagger]=0.
\label{211bis}\en

This list of  useful formulas is completed by
$\|P_lAP_s\|=\sqrt{x_s}\,\delta_{l+1,s}$ and by the following
commutation rule, \be
[A_L,A_L^\dagger]=\sum_{l=0}^L\,x_{l+1}\,\left(P_l-P_{l+1}\right)\,\Rightarrow\,
[A_L,A_L^\dagger]\Phi_k= \left\{
    \begin{array}{ll}
        x_1\,\Phi_0,\hspace{4.2cm}\mbox{ if } k=0,  \\
        \left(\chi_{L-k}\,x_{k+1}-\chi_{L+1-k}\,x_{k}\right)\,\Phi_{k}, \hspace{.4cm} \mbox{ if } k>0.\\
       \end{array}
        \right.\label{212}
\en This last equation, again, admits limit when
$L\rightarrow\infty$, and the limit is  \be [A_,A^\dagger]\Phi_k=
\left\{
    \begin{array}{ll}
        x_1\,\Phi_0,\hspace{4.3cm}\mbox{ if } k=0,  \\
        \left(x_{k+1}-x_{k}\right)\,\Phi_{k}, \hspace{2.7cm} \mbox{ if } k>0.\\
       \end{array}
        \right.\label{213}\en
Let us notice that, if $x_k=k$, i.e. for an harmonic oscillator,
this formula becomes $[A_,A^\dagger]\Phi_k=\Phi_k$ for all $k$,
which implies that $[A_,A^\dagger]=\I$, as it must be. In this
case a rigorous meaning can be given to the otherwise in general
formal expression
$[A_,A^\dagger]=\,x_1\,P_0+\sum_{l=1}^\infty\,(x_{l+1}-x_l)\,P_l$,
which is equal, in this case, to
$[A_,A^\dagger]=\sum_{l=0}^\infty\,P_l=\I$. However, this is not
the only case in which this expression is well defined: suppose
indeed that the following hold: $x_{l+1}\geq x_l$ for all
$l\in\N_0$ and
$\lim_{l\rightarrow\infty}\,x_l=\overline{x}<\infty$. In this case
we have
$$
\|\sum_{l=1}^\infty\,(x_{l+1}-x_l)\,P_l\|\leq
\sum_{l=1}^\infty\,(x_{l+1}-x_l)\,\|P_l\|=\sum_{l=1}^\infty\,(x_{l+1}-x_l)=(x_2-x_1)+
(x_3-x_2)+(x_4-x_3)+\ldots=
$$
$$=\lim_{l\rightarrow\infty}\,x_l-x_1=\overline{x}-x_1.$$
This implies that
$[A_,A^\dagger]=\,x_1\,P_0+\sum_{l=1}^\infty\,(x_{l+1}-x_l)\,P_l$
is well defined in the uniform topology, even if the sequence of
the $x_j$'s does not converge to zero. This is not a big surprise:
unbounded operators may have bounded commutators!

\vspace{2mm}

In Section IV we will make use of other properties of the
operators $A_L$ and $A_L^\dagger$. In particular using (\ref{23}),
 (\ref{28}) and (\ref{29}), it is possible to check that \bea \left\{
    \begin{array}{ll}
\left(A_L^\dagger\right)^2=\,\sum_{l=0}^{L-1}\,\sqrt{x_{l+1}\,x_{l+2}}\,P_{l+2,l}=\left(A^\dagger\right)^2\,Q_{L-1},\\
\left(A_L^\dagger\right)^3=\,\sum_{l=0}^{L-2}\,\sqrt{x_{l+1}\,x_{l+2}\,x_{l+3}}\,P_{l+3,l}
=\left(A^\dagger\right)^3\,Q_{L-2},\\
\ldots\ldots\ldots\\
\left(A_L^\dagger\right)^L=\,\sum_{l=0}^{1}\,\sqrt{x_{l+1}\,\ldots\,x_{l+L}}\,P_{l+L,l}
=\left(A^\dagger\right)^L\,Q_{1},\\
\left(A_L^\dagger\right)^{L+1}=\,\sqrt{x_{1}\,\ldots\,x_{L+1}}\,P_{L+1,0}=\left(A^\dagger\right)^{L+1}\,Q_{0},\\
\left(A_L^\dagger\right)^{L+2}=\left(A_L^\dagger\right)^{L+3}=\ldots=0.\\
       \end{array}
        \right.
 \label{214}\ena
This behavior is interesting because it extends, in a certain
sense, what happens for fermions: if $b$ and $b^\dagger$ are
fermion annihilation and creation operators, then it is well known
that $b^2=\left(b^\dagger\right)^2=0$. In this case, $A_L, A_L^2,
\ldots A_L^{L+1}$ are different from zero but all the other powers
are zero, and this result does not depend on the original choice
of $x_j$. It may be worth noticing, however, that
$\{A_L,A_L^\dagger\}=\I$ is in general not satisfied, not even in
an approximated form.

\section{A generalized harmonic oscillator and its algebraic dynamics}
The operators $A_L$ and $A_L^\dagger$ are closely related to those
introduced in \cite{bagJMP} in connection with a {\em standard
harmonic oscillator}. In this section we will repeat in some
details and giving more results the same analysis given in
\cite{bagJMP} for our generalized model. Of course, some
differences will arise at the beginning because of the different
 point of view we are considering here.

First of all let us remark that the o.n. basis of $\Hil$,
$\{\Phi_j\}$, belongs to the domain of all the relevant operators we
will consider here: $A$, $A^\dagger$, $A^\dagger\,A$,
$A\,A^\dagger$, as well as their {\em regularized} counterparts,
i.e. those obtained replacing the generic operator $X$ with
$X_L:=Q_{L+1}\,X\,Q_{L+1}$. Actually, they do much more than this:
they belong to the domain of all the powers of these operators.
Therefore we have, for instance, $\Phi_j\in\D:=
D^\infty(H_o):=\bigcap_{k\geq 0}\,D(H_o^k)$, for all $j$, where
$H_o=A^\dagger\, A$. The set $\D$ is dense in $\Hil$. As in
\cite{bagJMP} we introduce the *-algebra $\Lc^+(\D)$  of all the
closable operators defined on $\D$ which, together with their
adjoints, map $\D$ into itself. It is clear that all  powers of $A$
and $A^\dagger$ belong to this set. As we have already seen, $A_L$
and $A_L^+$ belong to $B(\Hil)$, but they also belong to
$\Lc^\dagger(\D)$ for any $L$.

In \cite{las}, the topological structures of both  $\D$ and
$\Lc^+(\D)$ are discussed in details; in particular, in $\D$ the
topology is defined by the following  seminorms: \be \phi \in \D
\rightarrow \|\phi\|_n\equiv \|H^n\phi\|, \label{31} \en where $n$
is a natural integer, $H$ is the closure of $H_o$ and $\|\, \|$ is
the norm of $\Hil$. The topology in $\Lc^+(\D)$ is introduced in
the following way. We start defining the set $\C$ of all the
positive, bounded and continuous functions $f(x)$ on $\R_+$, which
are decreasing faster than any inverse power of $x$:
$\sup_{x\in\mathbb{R}_+}\,f(x)\,x^n<\infty$ for each
$n\in\mathbb{N}_0$. The seminorms on $\Lc^+(\D)$ are labelled by
functions in $\C$ and by the integers $\N_0$. We have \be X\in
\Lc^+(\D) \rightarrow \|X\|^{f,k} \equiv
\max\left\{\|f(H)XH^k\|,\|H^kXf(H)\|\right\}. \label{32} \en Here
$\|\,\|$ is the usual norm in $B(\Hil)$. We use for this norm the
same notation as in equation (\ref{31}) since there is no
possibility of confusion. We observe that definition (\ref{32})
implies that $\|X\|^{f,k}=\|X^\dagger\|^{f,k}$ for each
$X\in\Lc^\dagger(\D)$, where $X^\dagger=X^*|_{\D}$. We call $\tau$
the topology on $\Lc^+(\D)$ defined by the seminorms (\ref{32}).
In \cite{las} it has been proved that $\Lc^+(\D)[\tau]$ is a
complete locally convex topological *-algebra.

Notice that the two contributions in the definition (\ref{32}) are
exactly of the same form. Therefore, the estimate of
$\|f(H)XH^k\|$ is very similar to the estimate of $\|H^kXf(H)\|$.
This is why, from now on, we will identify  $\|X\|^{f,k}$ simply
with $\|f(H)XH^k\|$.

With this in mind, and by means of the spectral decomposition for
$H$, $H=\sum_{l=1}^\infty x_l P_l$, we see that the seminorms can be
estimated as follows: \be X\in \Lc^+(\D) \longrightarrow \|X\|^{f,k}
\leq \sum_{l,s=1}^{\infty} f(x_l)x_s^k\|P_l\,X\,P_s\|. \label{33}
\en It is possible to check that, for all $X\in\Lc^+(\D)$, the
operator $X_{L-1}:=Q_{L}\,X\,Q_{L}$ belongs to $B(\Hil)$. Indeed we
have $\|X_{L-1}\|\leq
\|\,X\,Q_{L}\|=\sup_{\varphi\in\Hil\,\|\varphi\|=1}\,\|X\,Q_{L}\varphi\|<\infty$
since, for each fixed $L$ and for each $\varphi\in\Hil$,
$Q_{L}\varphi$ belongs to a finite dimensional Hilbert space. In
particular we have \be H_L=Q_{L+1}\,H\,Q_{L+1}=A_L^\dagger\,A_L\in
B(\Hil).\label{34} \en We can easily prove the following

\begin{lemma} Suppose that the sequence $\{x_s\}$ is increasing to $+\infty$ when $s\rightarrow\infty$. Then, for each
$l\in\mathbb{N}_0$,
$\lim_{L\rightarrow\infty}\|H^{-l}(\I-Q_L)\|=0$. Moreover, for
each $X\in\Lc^\dagger(\D)$,
$\tau-\lim_{L\rightarrow\infty}\,X_L=X$.
\end{lemma}
\begin{proof}
Since $H^{-l}=\sum_{s=1}^\infty x_s^{-l}\, P_s$, and
$\I-Q_L=\sum_{s=L+1}^\infty P_s$,  then
$H^{-l}(\I-Q_L)=\sum_{s=L+1}^\infty x_s^{-l}\,P_s$. Therefore, for
each $\varphi\in\Hil$,
$$\|H^{-l}(\I-Q_L)\varphi\|^2=\|H^{-l}(\I-Q_L)\,\varphi\|^2=
<\varphi,\,H^{-2l}(\I-Q_L)\varphi>=
$$
$$=\left(\sum_{s=L+1}^\infty
x_s^{-2l}\,<\varphi,\,P_s\,\varphi>\right)\leq
\frac{1}{x_{L+1}^{2l}}\,\sum_{s=0}^\infty
\,<\varphi,\,P_s\,\varphi>\leq\,
\frac{1}{x_{L+1}^{2l}}\,\|\varphi\|^2,
$$
so that the first statement follows from of our assumption on
$x_l$.

The proof of the second statement is very similar to that given in
\cite{bagtra}, and will not be repeated here.

\end{proof}

Let us now define, for each $X\in\Lc^\dagger(\D)$,
$\alpha_L^t(X)=\,e^{iH_L\,t}\,X\,e^{-iH_L\,t}$. This operator
satisfies the following, well defined, Heisenberg equation of
motion $\frac{d}{dt}\,\alpha_L^t(X)=i[H_L,\,\alpha_L^t(X)]$.
Notice that, on the contrary, the formal equation of motion
$\frac{d}{dt}\,\alpha^t(X)=i[H,\,\alpha^t(X)]$ needs not to be
well defined because of domain problems. However, it is easy to
check that the sequence $\{\alpha_L^t(X)\}$ define the algebraic
dynamics of $X$ in the following sense:
\begin{prop}\label{prop1}
Suppose that the sequence $\{x_s\}$ is such that $\exists
n_o\in\mathbb{N}$ such that $\{x_s^{-n_o}\}$ belongs to
$l^1(\mathbb{N})$. Then the sequences $\{e^{iH_L\,t}\}$ and
$\{\alpha_L^t(X)\}$ are both $\tau-$Cauchy in $\Lc^+(\D)$.

\end{prop}
\begin{proof}
Since $[H_L,H_M]=[H_L,H]=0$ for each $L$ and $M$, $M>L$, it is
easy to check that
$$I_{LM}:=\left\|f(H)\left(e^{i\,H_L\,t}-e^{i\,H_M\,t}\right)\,H^k\right\|=
2\,\left\|f(H)\sin\left(\frac{H_{ML}\,t}{2}\right)\,H^k\right\|,$$
where $H_{ML}=H_M-H_L$. Using the spectral decomposition of $H$ we
further get
$$
I_{LM}\leq
2\,\sum_{s=L+1}^M\,f(x_s)\,x_s^k\,\left|\sin\left(\frac{x_s\,t}{2}\right)\right|\,\|P_s\|\leq
2\,\sum_{s=L+1}^M\,f(x_s)\,x_s^k,
$$
which goes to zero when $M,L\rightarrow\infty$  for any divergent
sequence $\{x_n\}$ satisfying our assumption since $f(x)$ belongs
to $\C$. This is because surely there exists a positive constant
$c>0$ such that $f(x_s)\leq\frac{c}{x_s^{k+n_o}}$ for each
$s\in\mathbb{N}_0$, so that $I_{LM}\leq
2c\,\sum_{s=L+1}^M\,\frac{1}{x_s^{n_o}}$.

To prove now that $\{\alpha_L^t(X)\}$ is $\tau-$Cauchy for each
$X\in\Lc^\dagger(\D)$ we consider the  seminorm $
G_{LM}:=\|f(H)\left(\alpha_L^t(X)-\alpha_M^t(X)\right)\,H^k\|$
which, adding and subtracting $e^{iH_L\,t}\,X\,e^{-iH_M\,t}$, can
be estimated as follows
$$
G_{LM}\leq \|f(H)\,X
\left(e^{-iH_L\,t}-e^{-iH_M\,t}\right)\,H^k\|+
\|f(H)\left(e^{iH_L\,t}-e^{iH_M\,t}\right)\,X\,H^k\|\leq$$
$$
\leq
\|f(H)XH^{k+n_o}\|\,\|H^{-n_o}\left(e^{-iH_L\,t}-e^{-iH_M\,t}\right)\|+
\|H^{-n_o}\left(e^{iH_L\,t}-e^{iH_M\,t}\right)\|\,\|H^{n_o}f(H)XH^k\|,
$$
where $n_o$, at this stage, can be an arbitrary natural number. Of
course, both $\|f(H)XH^{k+n_o}\|$ and $\|H^{n_o}f(H)XH^k\|$ are
finite, since $X\in\Lc^\dagger(\D)$. We now choose the number
$n_o$ in such a way that $\|H^{-n_o}\left(e^{\pm iH_L\,t}-e^{\pm
iH_M\,t}\right)\|\rightarrow 0$ when $L,M\rightarrow\infty$. This
can be easily done when the sequence $\{x_s\}$ is given: it is
enough to take $n_o$ in such a way that $\{x_s^{-n_o}\}\in
l^1(\mathbb{N})$, $\sum_{s=1}^\infty\,x_s^{-n_o}<\infty$. Indeed,
with this choice and with the same estimates as above, we get for
instance
$$\|H^{-n_o}\left(e^{-iH_L\,t}-e^{-iH_M\,t}\right)\|=2\left\|
\sum_{s=L+1}^M\,x_s^{-n_o}\,\sin\left(\frac{x_s\,t}{2}\right)\,P_s\right\|\leq
\sum_{s=L+1}^M\,x_s^{-n_o}\rightarrow 0,
$$
when $M,L\rightarrow\infty$.

\end{proof}

This Proposition can be used to define both $e^{\pm iHt}$ and
$\alpha^t(X)$ for each $X\in\Lc^\dagger(\D)$: $e^{\pm
iHt}:=\tau-\lim_L\,e^{\pm iH_L\,t}$ and
$\alpha^t(X):=\tau-\lim_L\,e^{ iH_L\,t}\,X\,e^{- iH_L\,t}$, making
no use of intrinsically ill-defined series expansions for the
unbounded operator $H$. Of course, one could also try to give a
meaning to these quantities using the spectral theorem. However,
the procedure presented here seems promising for a possible
extension  to the situation in which no hamiltonian operator
exists, as for instance for mean field models, which are crucial
in concrete physical applications. A deeper analysis on this
aspect is in progress.

Another interesting remark concerns the following formula, which
relates $e^{iH_Lt}$ and $e^{iHt}$, which is obtained under the
assumption that all the quantities are well defined: \be
e^{iH_L\,t}=\I-Q_{L+1}+Q_{L+1}\,e^{iH\,t}\,Q_{L+1}.\label{35}\en
Indeed, since all the powers of $H$ belong to $\Lc^+(\D)$, while
$H_L^n$ belongs to $B(\Hil)$ for each $L$ and each $n$, it is
straightforward to check that $H_L^n=Q_{L+1}\,H^n\,Q_{L+1}$.
Therefore
$$
e^{iH_L\,t}=\I+iH_L\,t+\frac{1}{2!}\,(iH_L\,t)^2+\frac{1}{3!}\,(iH_L\,t)^3+\ldots=$$
$$=\I-Q_{L+1}+Q_{L+1}
\left(\I+iH\,t+\frac{1}{2!}\,(iH\,t)^2+\frac{1}{3!}\,(iH\,t)^3+\ldots\right)Q_{L+1},
$$
which returns (\ref{35}) if the series above converges. In
general, therefore, (\ref{35}) must be seen only as a formal
equality.

\vspace{2mm} Let us now consider the dynamics at the infinitesimal
level, that is at the level of the derivations. This kind of
problems has been recently analyzed in \cite{bit} and \cite{bit2}.
The first obvious remark is that, since $\tau-\lim_{L}\,H_L^k=H^k$
for all $k\in\N_o$, and since the multiplication is separately
continuous in the topology $\tau$, the {\em regularized}
derivation $\delta_L(X)=i[H_L,X]$ converges to the derivation
$\delta(X)=i[H,X]$ in $\tau$, for each operator
$X\in\Lc^\dagger(\D)$. However, this is not enough to conclude
that also the repeated commutators do converge, since these
contain contributions like $H_L^n\,X\,H_L^m$. Nevertheless,
introducing recursively the $k-th$ regularized derivation as
$\delta_L^k(X)=i[H_L,\delta_L^{k-1}(X)]$, the following
Proposition holds true:
\begin{prop}
For each $k\geq1$ and for each $X\in\Lc^\dagger(\D)$,
$\tau-\lim_{L}\,\delta_L^k(X)$ exists in $\Lc^\dagger(\D)$, and it
defines an operator which we call $\delta^{(k)}(X)$, and which
obeys the following recursion formula:
$\delta^{(k)}(X)=i[H,\delta^{(k-1)}(X)]$.

\end{prop}
\begin{proof}
Our claim is proved by induction. The statement is clearly true
for $k=1$. Suppose therefore it holds for a fixed $k$, that is
that
$\tau-\lim_{L}\,\delta_L^k(X)=\delta^{(k)}(X)\in\Lc^\dagger(\D)$.
This means that $\delta^{(k+1)}(X)=i[H,\delta^{(k)}(X)]$ is well
defined in $\Lc^\dagger(\D)$. Then we have
$$
\|f(H)\left(\delta_L^{k+1}(X)-\delta^{(k+1)}(X)\right)\,H^l\|\leq$$
$$\leq
\|f(H)\left(H_L\,\delta_L^{k}(X)-H\,\delta^{(k)}(X)\right)\,H^l\|+
\|f(H)\left(\delta_L^{k}(X)\,H_L-\delta^{(k)}(X)\,H\right)\,H^l\|=:P_1^L+P_2^L,
$$
with obvious notation. Adding and subtracting $H\,\delta_L^{k}(X)$
in $P_1^L$ we get
$$
P_1^L\leq
\|f(H)(H_L-H)\,\delta_L^{k}(X)\,H^l\|+\|f(H)H\,\left(\delta_L^{(k)}(X)-\delta^{(k)}(X)\right)\,H^l\|
=:P_{1,A}^L+P_{1,B}^L
$$
again with obvious notation. Since $x\,f(x)\in\C$ for each
$f(x)\in \C$, our induction assumption easily implies that
$\lim_{L}\,P_{1,B}^L=0$. As for $P_{1,A}^L$, recalling that
$[H,H_L]=[H,Q_L]=0$, we deduce that
$$
P_{1,A}^L\leq
\|H^{-1}\,(Q_L-\I)\|\,\|f(H)H^2\,\delta_L^k(X)\,H^l\|\rightarrow 0
$$
when $L\rightarrow\infty$ using Lemma 1 above, since
$x^2\,f(x)\in\C$ and since $\delta_L^k(X)$, being
$\tau$-convergent, is necessarily $\tau$-bounded. Therefore we
have $\lim_{L}\,P_{1}^L=0$. The proof of $\lim_{L}\,P_{2}^L=0$
follows essentially the same steps.

\end{proof}

It is worth remarking that we have used here, as in \cite{bagtra},
the notation $\delta^{(k)}(X)$ instead of a maybe more natural
$\delta^k(X)$ since this cannot in general be defined, as one
expects, as $\delta^k(X)=i[H,\delta^{k-1}(X)]$, because the rhs of
this equation is not everywhere defined when $H$ is unbounded.

What this Proposition incidentally proves is that the series in
$\alpha_L^t(X)=\sum_{k=0}^\infty\,\frac{t^k}{k!}\,\delta_L^k(X)$
is term to term $\tau$-convergent to
$\sum_{k=0}^\infty\,\frac{t^k}{k!}\,\delta^{(k)}(X)$, which could
be used to  define $\alpha^t(X)$. Using our previous results we
can therefore state the following: for each $X\in\Lc^\dagger(\D)$
and for each $t\in\mathbb{R}$, \be
e^{iHt}\,X\,e^{-iHt}=\alpha^t(X)=\tau-\lim_{L\rightarrow\infty}\,\alpha_L^t(X)=
\tau-\lim_{L\rightarrow\infty}\,\sum_{k=0}^\infty\,\frac{t^k}{k!}\,\delta_L^k(X)\label{37}\en

\section{Approximated coherent states}
This section is devoted to the possibility of using the operators
$A_L$ and $A_L^\dagger$ to generate {\em approximate} coherent
states (ACS), that is vectors, depending on the regularizing
parameter $L$, which share all the properties of the coherent
states in the limit $L\rightarrow\infty$ and an approximated
version of these  if $L$ is left finite.

In the old literature, see \cite{ks} for instance, a standard
coherent state (SCS) is a vector arising from the action of the
unitary operator $U(z)=e^{z\,a^\dagger-\overline{z}\,a}$,
$z\in\mathbb{C}$ and $[a,a^\dagger]=\I$, on the vacuum of $a$,
$\Phi_0$, $a\Phi_0=0$: $|z>=U(z)\Phi_0$. These normalized vectors
can be written in other equivalent ways, introducing the o.n.
basis $\{\Phi_n,\,n\in\mathbb{N}_0\}$ where
$\Phi_n=\frac{(a^\dagger)^n}{\sqrt{n!}}\,\Phi_0$, as follows: \be
|z>=U(z)\Phi_0=e^{-|z|^2/2}e^{z\,a^\dagger}\,\Phi_0=e^{-|z|^2/2}\,\sum_{k=0}^\infty
\,\frac{z^n}{\sqrt{n!}}\,\Phi_n.\label{51}\en They share a lot of
interesting properties, among which the most interesting for us
are the following:\begin{enumerate} \item $|z>$ is an eigenstate
of $a$: $a|z>=z|z>$.\item They satisfy a {\em resolution of the
identity}: $\frac{1}{\pi}\,\int\,d^2z\,|z><z|=\I$.\item They
saturate the Heisenberg uncertainty principle: let
$q=\frac{a+a^\dagger}{\sqrt{2}}$,
$p=\frac{a-a^\dagger}{i\,\sqrt{2}}$, $(\Delta X)^2=<X^2>-<X>^2$
for $X=q,p$, then $\Delta q\,\Delta p=\frac{1}{2}$.
\end{enumerate}
These properties are recovered using a different definition for
the coherent states, see \cite{ali} and references therein,
definition which generalizes the one above and which appears
strictly related to the procedure introduced in Section II.
Starting from a sequence $\{x_l,\,l\in\N_0\}$, of non negative
numbers, $x_l\geq 0$ for all $l\in\N_0$, it is possible to define
some vectors, parametrized by a complex $z$,  as follows: \be
\Xi(z):=N(|z|^2)^{-1/2}\,\sum_{k=0}^\infty
\,\frac{z^n}{\sqrt{x_n!}}\,\Phi_n,\label{52}\en where
$N(|z|^2)=\sum_{k=0}^\infty \,\frac{|z|^{2n}}{x_n!}$, so that
$<\Xi,\Xi>=1$ for all $|z|\leq\rho$, $\rho$ being the radius of
convergence of the series for $N$, and where $x_0!=1$ and
$x_n!=x_1\,x_2\ldots x_n$. In particular, if $A$ is the operator
in (\ref{27}), we have $A\,\Xi(z)=z\,\Xi(z)$, so that these
generalized coherent states (GCS) are again eigenstates of the
generalized annihilation operator $A$. Moreover, in \cite{ali} it
is also shown that the existence of a resolution of the identity,
that is the existence of a measure $d\nu(z,\overline{z})$ such
that
$\int_{C_\rho}\,N(|z|^2)\,|\Xi(z)><\Xi(z)|\,d\nu(z,\overline{z})=\I$,
 is related to the existence of a solution of the
following {\em moment problem}: we put $z=r\,e^{i\theta}$,
$d\nu(z,\overline{z})=d\theta\,d\lambda(r)$,
$C_\rho=\{z=r\,e^{i\theta}, \,\theta\in[0,2\pi[, r\in[0,\rho[\}$,
then we want $d\lambda(r)$ to be such that \be
\int_0^{\rho}\,d\lambda(r)\,r^{2k}=\frac{x_k!}{2\pi}, \quad
\forall k\in\N_0. \label{53}\en It is known that this problem has
not always solution, but when it does, then a resolution of the
identity can be established.

Finally, if we introduce two self-adjoint operators (the
generalized position and momentum operators)
$Q=\frac{A+A^\dagger}{\sqrt{2}}$,
$P=\frac{A-A^\dagger}{i\,\sqrt{2}}$, then $\Xi(z)$ saturates again
the Heisenberg uncertainty principle, which now can be written as
\be \Delta Q\,\Delta
P=\frac{1}{2}\left|<AA^\dagger>-|z|^2\right|.\label{54}\en Notice
that, if $x_k=k$, and therefore $A=a$, this returns the standard
expression: $\Delta Q\,\Delta P=\frac{1}{2}$.

As we observe, the three possible equivalent definitions for the
SCS are replaced by an unique definition for the GCS. This is due
to the fact that, since in general $[A,A^\dagger]\neq\I$, the
 equality $e^{z\,a^\dagger-\overline{z}\,a}=e^{-|z|^2/2}\,e^{z\,a^\dagger}\,
 e^{-\overline{z}\,a}$ does not extend any-longer: $e^{z\,A^\dagger-\overline{z}\,A}\neq
 e^{-|z|^2/2}\,e^{z\,A^\dagger}\,
 e^{-\overline{z}\,A}$. More generalizations can be found in \cite{aag}, but
 they will have no role along this paper. Of course, since $e^{-|z|^2/2}\,\sum_{k=0}^\infty
\,\frac{z^n}{\sqrt{n!}}\,\Phi_n$ does not explicitly refer to the
creation and annihilation operators, this can be naturally
generalized as in (\ref{52}). We want now to introduce our ACS,
which should be defined starting from
$A_L^\sharp=Q_{L+1}\,A^\sharp\,Q_{L+1}$, and we discuss their
relations with GCS.

For that we define, starting from the fixed sequence
$\{x_l,\,l\in\N_0\}$, the analytic function
$F(z)=\sum_{k=0}^\infty \,\frac{z^n}{x_n!}$, whose radius of
convergence is exactly $\rho$. We use $F(z)$ to define our ACS in
the following way: \be
\Psi_L(z):=N_{\Psi_L}(|z|^2)^{-1/2}\,F(zA_L^\dagger)\,\Phi_0\label{55}\en
where $N_{\Psi_L}(|z|^2)$ must be chosen in such a way that
$<\Psi_L(z),\Psi_L(z)>=1$ for all $L$ and for each $z$ inside the
domain of convergence.

If, for instance, $x_n=n$, then $F(z)=e^z$. In this case we claim
that $\lim_{L\rightarrow\infty}\Psi_L(z)=|z>$, that is we recover
the SCS. This statement will be proved in the following. Here we
want to prove that $\Psi_L(z)$ can be written in another form,
which is more closely related to definition (\ref{52}) and which
is useful for further considerations. The starting point for this
is the result in (\ref{214}) which states, in particular, that
$\left(A_L^\dagger\right)^j=0$ for $j=L+2, L+3,\ldots$, and that
$A_L^\dagger\Phi_0=\sqrt{x_1}\,\Phi_1$,
$\left(A_L^\dagger\right)^2\Phi_0=\sqrt{x_1\,x_2}\,\Phi_2,
\ldots$, $\left(A_L^\dagger\right)^{L+1}\Phi_0=\sqrt{x_1\ldots
x_{L+1}}\,\Phi_{L+1}$. Therefore we have
$$
\Psi_L(z):=N_{\Psi_L}(|z|^2)^{-1/2}\,F(zA_L^\dagger)\,\Phi_0=N_{\Psi_L}(|z|^2)^{-1/2}\,
\sum_{k=0}^\infty\,\frac{\left(zA_L^\dagger\right)^k}{x_k!}\,\Phi_0=$$
$$=N_{\Psi_L}(|z|^2)^{-1/2}\,
\sum_{k=0}^{L+1}\,\frac{\left(zA_L^\dagger\right)^k}{x_k!}\,\Phi_0,
$$
which gives the following alternative expression for $\Psi_L(z)$,
\be
\Psi_L(z):=N_{\Psi_L}(|z|^2)^{-1/2}\,\sum_{k=0}^{L+1}\,\frac{z^k}{\sqrt{
x_k! }}\,\Phi_k. \label{56}\en This formula shows that $\Psi_L(z)$
is a finite linear combination of the $\Phi_j$'s, and, recalling
(\ref{52}), is also  the most natural one: it can be obtained from
(\ref{52}) simply restricting the series to the first $L+1$ terms
and, as a consequence, modifying  also the normalization, which
must be chosen as
$N_{\Psi_L}(|z|^2)=\sum_{k=0}^{L+1}\,\frac{|z|^{2k}}{{ x_k! }}$.
Recalling definitions (\ref{21}) and (\ref{52}) we can easily
relate $\Psi_L(z)$ with $\Xi(z)$: \be
\Psi_L(z)=\sqrt{\frac{N(|z|^2)}{N_{\Psi_L}(|z|^2)}}\,Q_{L+1}\,\Xi(z),\label{56bis}\en
which, of course, indicates that
$\lim_{L\rightarrow\infty}\Psi_L(z)=\Xi(z)$, as expected.

Formulas (\ref{55}) and (\ref{56}) are, in a certain sense, our
counterparts of (\ref{51}) for these ACS: they extend the second
and the third possibilities in (\ref{51}) to the case in which
$x_n$ is generic. Notice that, if $x_n=n$, then we recover  a sort
of natural cutoff of the SCS.

The ACS $\Psi_L(z)$ turns out to be an approximated eigenstate of
$A_L$. Indeed we have \be A_L
\Psi_L(z)=z\,\sqrt{\frac{N_{\Psi_{L-1}}(|z|^2)}{N_{\Psi_{L}}(|z|^2)}}\,\Psi_{L-1}(z),\label{57}\en
which formally converges, when $L\rightarrow\infty$ to the
eigenvalue equation $A\Xi(z)=z\,\Xi(z)$. Let us first prove
equation (\ref{57}), and then we will discuss in more details what
happens in the limit $L\rightarrow\infty$.

We start remarking that, for each fixed $L$, $A_L\Phi_0=0$,
$A_L\,\Phi_k=\sqrt{x_k}\,\Phi_{k-1}$, for $k=1,2,\ldots,L+1$, and
$A_L\,\Phi_k=0$ for $k\geq L+2$. Therefore we have
$$
A_L
\Psi_L(z)=N_{\Psi_L}(|z|^2)^{-1/2}\,\sum_{k=0}^{L+1}\,\frac{z^k}{\sqrt{
x_k! }}\,
A_L\Phi_k=N_{\Psi_L}(|z|^2)^{-1/2}\,\sum_{k=1}^{L+1}\,\frac{z^k}{\sqrt{
x_k! }}\, \sqrt{x_k}\,\Phi_{k-1},
$$
so that formula (\ref{57}) immediately follows.

As for the limit for $L$ divergent, it is possible to check that
$$
\lim_{L\rightarrow\infty}\,\|A_L
\Psi_L(z)-A\Xi(z)\|=\lim_{L\rightarrow\infty}\,\left\|z\,
\sqrt{\frac{N_{\Psi_{L-1}}(|z|^2)}{N_{\Psi_{L}}(|z|^2)}}\,
\Psi_{L-1}(z)-z\Xi(z)\right\|=0.
$$
We leave the proof of this statement to the reader. Here we just
notice that this shows that (\ref{57}) is an approximated version
of the eigenvalue equation $A\Xi(z)=z\,\Xi(z)$.

We can also check that the states $\Psi_L(z)$ generate an
approximated decomposition of the identity. Indeed, if we take
$d\nu(z,\overline{z})$ as before, with $d\lambda(r)$ satisfying
the moment problem in (\ref{53}), with exactly the same strategy
we can check that
$\int_{C_\rho}\,N_{\Psi_{L}}(|z|^2)\,|\Psi_L(z)><\Psi_L(z)|\,d\nu(z,\overline{z})=
\sum_{k=0}^{L+1}\,P_k=Q_{L+1}$, which is not $\I$ but  tends to
$\I$ when $L\rightarrow\infty$ in the strong topology.

The final step of this analysis is the uncertainty relation.
Again, we expect that this is saturated only in an approximated
version and this is indeed what happens. Defining as usual
$Q_L=\frac{A_L+A_L^\dagger}{\sqrt{2}}$ and
$P_L=\frac{A_L-A_L^\dagger}{i\,\sqrt{2}}$, and repeating the usual
computations, we find that \be \Delta Q_L\,\Delta
P_l=\frac{1}{2}\,\sqrt{\Gamma_{2,L}^2-\Gamma_{1,L}^2},\label{58}\en
where
$$\Gamma_{1,L}=<A_L\,A_L^\dagger>+|z|^2\,\frac{N_{\Psi_{L-2}}(|z|^2)}{N_{\Psi_{L-1}}(|z|^2)}
\left(1-2\,\frac{N_{\Psi_{L-2}}(|z|^2)}{N_{\Psi_{L-1}}(|z|^2)}\right),$$
and
$$\Gamma_{2,L}=(z^2+\overline{z}^2)\,\frac{1}{N_{\Psi_{L-1}}(|z|^2)}
\left(N_{\Psi_{L-1}}(|z|^2)-\,\frac{N_{\Psi_{L-2}}(|z|^2)}{N_{\Psi_{L-1}}(|z|^2)}\right).$$
To check whether the Heisenberg inequality is saturated, that is
if $\Delta Q_L\,\Delta P_l=\frac{1}{2}\,(<[A_L,A_L^\dagger]>)$, we
should compute the rhs, which gives
$\frac{1}{2}\,\left(<A_L\,A_L^\dagger>-|z|^2\,\frac{N_{\Psi_{L-2}}(|z|^2)}{N_{\Psi_{L-1}}(|z|^2)}\right)$.
We see that, for $L<\infty$, in general $\Delta Q_L\,\Delta
P_l=\frac{1}{2}\,\sqrt{\Gamma_{2,L}^2-\Gamma_{1,L}^2}\geq
\frac{1}{2}\,\left(<A_L\,A_L^\dagger>-|z|^2\,\frac{N_{\Psi_{L-2}}(|z|^2)}{N_{\Psi_{L-1}}(|z|^2)}\right)$.
However,  since $\Gamma_{1,L}\rightarrow <AA^\dagger>- |z|^2$ and
$\Gamma_{2,L}\rightarrow 0$, when $L\rightarrow\infty$, then
$\Delta Q_L\,\Delta
P_l\rightarrow\frac{1}{2}\left|<AA^\dagger>-|z|^2\right|$ which
coincides with the limit of
$\frac{1}{2}\,\left(<A_L\,A_L^\dagger>-|z|^2\,\frac{N_{\Psi_{L-2}}(|z|^2)}{N_{\Psi_{L-1}}(|z|^2)}\right)$
for $L$ diverging . In other words, $\Psi_L(z)$ saturates the
Heisenberg uncertainty relation in the limit $L\rightarrow\infty$.

\vspace{2mm}

We want now to find the explicit expression of a vector in $\Hil$
in order that it is an eigenstate of $A_L$. This is a sort of {\em
inverse problem} of the one considered so far.

Since $\Phi_j$ is an o.n. basis of $\Hil$, if such an eigenstate
$\Upsilon_L(z)$ exists it must admit an expansion like
$\Upsilon_L(z)=\sum_{k=0}^\infty \,b_k^{(L)}(z)\,\Phi_k$, and the
problem consists in finding the coefficients $b_k^{(L)}(z)$ in
such a way that $A_L\,\Upsilon_L(z)=z\,\Upsilon_L(z)$ is
satisfied. However, it is not difficult to check that this
requirement implies that $b_j^{(L)}(z)=0$ for $j=0,1,2,\ldots$, so
that no nontrivial solution of the eigenvalue equation can exist.
For this reason we weaken the requirement, following a suggestion
already contained in formula (\ref{57}): we look for an {\em
approximate eigenstate} of $A_L$, that is a vector of $\Hil$ such
that $A_L\,\Upsilon_L(z)=z\,\Upsilon_{L-1}(z)$. It is convenient
to expand $\Upsilon_L(z)$ as it follows \be
\Upsilon_L(z)=\sum_{k=0}^\infty
\,b_k^{(L+1)}(z)\,\Phi_k.\label{59}\en We need to compute the
coefficients of this expansion. In order to satisfy the equality
$A_L\Upsilon(z)_L=z\Upsilon(z)_{L-1}$ these must obey the
following recursion formula: \be
b_j^{(L+1)}(z)\,\sqrt{x_j}=z\,b_{j-1}^{(L)}(z)\label{510}\en with
$b_j^{(L+1)}(z)=0$ for each $j>L+1$. This implies that
$\Upsilon_L(z)$  must be of the following form \be
\Upsilon_L(z)=\sum_{k=0}^{L+1}\,b_0^{(L+1-k)}(z)\,\frac{z^k}{\sqrt{x_k!}}\,\Phi_k\label{511}\en
for any choice of $b_0^{(j)}(z)$, which however should  be
conveniently chosen if we also want $\Upsilon_L(z)$ to be
normalized.

If we take for instance
$b_0^{(0)}(z)=b_0^{(1)}(z)=\ldots=b_0^{(L+1)}(z)=b(z)$, then we
get
$\Upsilon_L(z)=b(z)\,\sum_{k=0}^{L+1}\,\frac{z^k}{\sqrt{x_k!}}\,\Phi_k$,
which coincides with (\ref{56}) but for the normalization
constant, which can be computed easily, recovering exactly the
result in (\ref{56}).

Another interesting choice of $b_0^{(j)}(z)$ is the following:
$b_0^{(j)}(z)=\frac{z^j}{\sqrt{x_j!}}$, which produces the
following ACS:
$\Upsilon_L(z)=z^{L+1}\,\sum_{k=0}^{L+1}\,\frac{1}{\sqrt{x_{L+1-k}!\,x_k!}}\,\Phi_k$.
Here, as it is evident, $z$ appears only to the single power
$L+1$. Notice that, again, $\Upsilon_L(z)$ is not normalized. In
general, normalizing such a $\Upsilon_L(z)$ produce a different
function $\hat\Upsilon_L(z)=N_L^{-1/2}(z)\,\Upsilon_L(z)$, which
satisfies
$A_L\,\hat\Upsilon_L(z)=z\,\sqrt{\frac{N_{L-1}(z)}{N_L(z)}}\,\hat\Upsilon_{L-1}(z)$.
We see that this does not differ significantly from
$A_L\,\hat\Upsilon_L(z)=z\,\hat\Upsilon_{L-1}(z)$ for large $L$,
since $N_{L-1}(z)\simeq N_L(z)$ in this case.

Another possible choice that we have considered is
$b_0^{(j)}(z)=\sqrt{\frac{x_j!}{j!}}$. It is not hard to imagine
more choices, some of which could be useful in concrete
applications. In general the examples already considered show that
different choices of $b_0^{(j)}$ produce ACS with different
analytical and, possibly, physical characteristics.

Of course, the choice of the coefficients produces, in turns, the
related moment problem (\ref{53}) which must be satisfied in order
to get an approximated resolution of the identity.

\vspace{2mm}

{\bf Remark:} it is well known that coherent states are deeply
connected with squeezed states, so that one could try to repeat
the same analysis considered here for these states. This will be
discussed in a forthcoming paper.

\vspace{2mm}

We end this section considering a different kind of coherent
states, those first introduced by Gazeau and Klauder, \cite{GK},
and generalized in \cite{ba}. These states, labelled by two real
numbers $J$ and $\gamma$, can be written in terms of the o.n.
basis of a self-adjoint operator $H=H^\dagger$, $|n>$, as \be
|J,\gamma>=N(J)^{-1}\,\sum_{n=0}^\infty\,\frac{\,J^{n/2}\,e^{-i\epsilon_n\,\gamma}}{\sqrt{\rho_n}}\,|n>,\label{512}\en
where $N(J)^2=\sum_{n=0}^\infty\,\frac{\,J^{n}\,}{\rho_n}$,
$H|n>=\omega\,\epsilon_n\,|n>$, with
$0=\epsilon_0<\epsilon_1<\epsilon_2<\ldots$. It may be worth
noticing that the normalization of these GK-states is $N(J)^{-1}$
and formally differs from the one used by Ali et al, \cite{ali}.
We adopt here the same notation as in the original papers. These
states satisfy the following properties:
\begin{enumerate}
\item if there exists a non negative function, $\rho(u)$, such
that $\int_0^R\,\rho(u)\,u^n\,du=\rho_n$ for all $n\geq 0$, where
$R$ is the radius of convergence of $N(J)$, then, introducing a
measure $d\nu(J,\gamma)=N(J)^2\,\rho(J)\,dJ\,d\nu(\gamma)$, with
$\int_{\mathbb{R}}\ldots
\,d\nu(\gamma)=\lim_{\Gamma\rightarrow\infty}\,\frac{1}{2\Gamma}\,
\int_{-\Gamma}^\Gamma\ldots\,d\gamma$, the following resolution of
the identity is satisfied: \be
\int_{C_R}\,d\nu(J,\gamma)\,|J,\gamma><J,\gamma|=\int_0^R\,N(J)^2\,\rho(J)\,dJ\,
\int_{\mathbb{R}}\,d\nu(\gamma)\,|J,\gamma><J,\gamma|=\I;
\label{513}\en \item the states $|J,\gamma>$ are {\em temporarily
stable}: \be e^{-iHt}\,|J,\gamma>=|J,\gamma+\omega t>,\quad \forall
t\in\mathbb{R};\label{514}\en \item if $\rho_n=x_n!$ then they
satisfy the {\em action identity}: \be
<J,\gamma|H|J,\gamma>=J\,\omega;\label{515}\en
\item they are continuous: if $(J,\gamma)\rightarrow(J_0,\gamma_0)$
then $\||J,\gamma>-|J_0,\gamma_0>\|\rightarrow 0$.
\end{enumerate}
It is interesting to observe that the states $|J,\gamma>$ are
eigenstates of the following $\gamma-$ depending annihilation-like
operator $a_\gamma$ defined on $|n>$ as follows: \be
a_\gamma\,|n>=\left\{
    \begin{array}{ll}
        0,\hspace{4.4cm}\mbox{ if } n=0,  \\
        \sqrt{\epsilon_n}\,e^{i(\epsilon_n-\epsilon_{n-1})\,\gamma}|n-1>, \hspace{0.6cm} \mbox{ if } n>0,\\
       \end{array}
        \right.
\label{516}\en whose adjoint acts as
$a_\gamma^\dagger\,|n>=\sqrt{\epsilon_{n+1}}\,e^{-i(\epsilon_{n+1}-\epsilon_{n})\,\gamma}|n+1>$.
This shows that  $H$ can be written as
$H=\omega\,a_\gamma^\dagger\,a_\gamma$. With standard computations
we can also check that \be a_\gamma |J,\gamma>=\sqrt{J}\,
|J,\gamma>.\label{517}\en However, it should be stressed that
$|J,\gamma>$ is not an eigenstate of $a_{\gamma'}$ if
$\gamma\neq\gamma'$.

Using the suggestion coming from formula (\ref{56bis}) we define
new vectors $|J,\gamma;L>$ as \be
|J,\gamma;L>=\frac{N(J)}{N_L(J)}\,Q_{L+1}\,|J,\gamma>,\label{518}\en
where as usual $Q_{L+1}=\sum_{k=0}^{L+1}\,|k><k|$ and where
$N_L(J)^2=\sum_{k=0}^{L+1}\,\frac{J^n}{\rho_n}$. These states are
interesting, since they satisfy the following
properties:\begin{enumerate}\item they can be written as \be
|J,\gamma;L>=\frac{1}{N_L(J)}\,\sum_{n=0}^{L+1}\,\frac{\,J^{n/2}\,e^{-i\epsilon_n\,
\gamma}}{\sqrt{\rho_n}}\,|n>,\label{519}\en \item  if there exists
a function, $\rho(u)$ with the same features as above, then,
introducing a measure
$d\nu_L(J,\gamma)=N_L(J)^2\,\rho(J)\,dJ\,d\nu(\gamma)$, with
$d\nu(\gamma)$ as before, the following  identity holds true: \be
\int_{C_R}\,d\nu_L(J,\gamma)\,|J,\gamma;L><J,\gamma;L|=Q_{L+1}
\label{520}\en \item the states $|J,\gamma;L>$ are temporarily
stable: \be e^{-iHt}\,|J,\gamma;L>=|J,\gamma+\omega
t;L>\quad\forall t\in\mathbb{R},\,\forall L;\label{521}\en \item
they satisfy the following identity: \be
<J,\gamma;L|H_L|J,\gamma;L>=J\,\omega\,\left(\frac{N_{L-1}(J)}{N_L(J)}\right)^2;\label{522}\en
\item they are continuous: if
$(J,\gamma)\rightarrow(J_0,\gamma_0)$ then
$\||J,\gamma;L>-|J_0,\gamma_0;L>\|\rightarrow 0$.
\end{enumerate}
The proofs of these statements are not significantly different
from that of the original ones and will be omitted here. As we
see, the states $|J,\gamma;L>$  satisfy, mostly in an approximated
version, the requirement of the GK-coherent states. A relevant
feature is that they satisfy temporal stability exactly for each
value of $L$, even if they differ from the original GK-states
since they are just a finite linear combinations of the
eigenstates $|n>$ of $H$.

\section{Connections with quons}

In a series of papers \cite{moh,fivetc} many people have
introduced and analyzed a different kind of {\em elementary
particles}, the so called quons. They interpolate between bosons
and fermions, in the sense that they satisfy a modified version of
the canonical commutation relations which depend on a parameter
$q$. More explicitly, when $q=1$ these particles obey CCR while,
if $q=-1$, they obey CAR. Different kind of quons have been
proposed in the literature but maybe the most common are those
satisfying the following {\em q-mutators}: \be
a\,a^\dagger-q\,a^\dagger\,a=\I \quad\mbox{ or }\quad
a\,a^\dagger-q\,a^\dagger\,a=q^{-2{\hat N}}\,\I,\label{41}\en
where $\hat N$ is the number operator: $\hat N\Phi_j=j\Phi_j$,
\cite{moh}. We will call respectively {\em first} and {\em second
kind quons} those satisfying the first or the second q-mutator
above. In this section we will discuss briefly the relation
between these quons, and other obeying other generalized
q-mutation relations, with the operators $A$ and $A^\dagger$
introduced in Section II, extending some results first discussed
by Ali et al., \cite{ali2}.

The spectral decompositions for the operators $A^\dagger\,A$ and
$A\,A^\dagger$ are respectively
$\sum_{l=0}^\infty\,x_{l+1}\,P_{l+1}$ and
$\sum_{l=0}^\infty\,x_{l+1}\,P_l$, which implies that
$$A\,A^\dagger-q\,A^\dagger\,A=x_1\,P_0+(x_2-qx_1)\,P_1+(x_3-qx_2)\,P_2+\ldots.$$
This expansion is equal to the identity operator
$\I=\sum_{l=0}^\infty\,P_{l}$ if and only if $x_1=1$,
$x_2-qx_1=x_3-qx_2=x_4-qx_3=\ldots=1$. Therefore we find
 \be x_{n+1}=1+q+q^2+\ldots+q^n=
\left\{
    \begin{array}{ll}
        n+1,\hspace{2.8cm}\mbox{ if } q=1,  \\
        \frac{1-q^{n+1}}{1-q}, \hspace{2.6cm} \mbox{ if } q\neq 1.\\
       \end{array}
        \right.
\label{42} \en It is worth noticing that the result $x_n=n$ when
$q=1$ is expected: in this case, indeed, the q-mutator rules in
(\ref{41}), as well as the ones we will consider below, return the
CCR. Let us also remind that
$H_o\Phi_n=A^\dagger\,A\,\Phi_n=x_n\,\Phi_n$. This means that, if
we are interested in using the results obtained in the previous
sections to the analysis of first kind quons, then the sequence
$\{x_j\}$ must be chosen as in (\ref{42}). We notice also that
this equation, if $q\notin[-1,1[$, produce unbounded operators. In
other words, we could use the results given in Section III and IV
to study the algebraic dynamics of free quons and the ACS
associated to them.

A standard problem in this topic is related to the definition of
the  number operator $\hat N$. For first kind quons this operator
satisfies the equation $H_o=\frac{1-q^{\hat N}}{1-q}$, \cite{moh},
and can be written as $\hat
N=\frac{1}{\log(q)}\,\log(\I-H_o(1-q))$.

The same analysis can be carried out for the second kind quons. In
this case the recursive formula for $x_n$ produces  \be
x_{n+1}=q^n\,\sum_{k=0}^n\,\left(\frac{1}{q^3}\right)^k= \left\{
    \begin{array}{ll}
        n+1,\hspace{4.0cm}\mbox{ if } q=1,  \\
        q^n\,\frac{1-(1/q^3)^{n+1}}{1-1/q^3}, \hspace{2.7cm} \mbox{ if } q\neq 1.\\
       \end{array}
        \right.
\label{43} \en Even in this case a number operator $\hat N$ can be
introduced. It must satisfy the following relation $H_o=q^{{\hat N
}-1}\,\frac{1-(1/q^3)^{{\hat N}}}{1-1/q^3}$. Such an operator
satisfies again $\hat N\Phi_j=j\,\Phi_j$, \cite{moh}.

These conclusions can be easily generalized to different
expressions of q-mutators which are not usually considered in the
literature: in general, given $H_o$ satisfying the eigenvalue
equation $H_o\,\Phi_j=x_j\,\Phi_j$, and introducing a map
$X:\N_0\rightarrow\R$ such that, for each $m\in\N_0$ $X(m)=x_m$,
the number operator $\hat N$ is related to $H_o$ by $H_o=X(\hat
N)$, which, if $X$ is invertible, gives $\hat N$ as a function of
$H_o$: $\hat N=X^{-1}(H_o)$.

Let us now introduce the following very general q-mutation
relation: \be a\,a^\dagger-q\,a^\dagger\,a=f(q,\hat
N),\label{44}\en where $f(q,\hat N)$ is a self-adjoint operator
depending on $q$ and $\hat N$ and such that
$\lim_{q\rightarrow\pm1}\,f(q,N)=1$ in some {\em topological}
sense, for instance strongly on a dense domain. Using the spectral
decomposition $f(q,\hat N)=\sum_{j=0}^\infty\,f(q,j)\,P_j$, and
repeating the same steps as before, we get the following
expression for the $x_n$'s: \be x_{n+1}=q^n\,f(q,0)+
q^{n-1}\,f(q,1)+q^{n-2}\,f(q,2)+\ldots+q\,f(q,n-1)+f(q,n)\label{45}\en
It is an easy exercise to check that this formula returns the
results already obtained for the first and the second kind quons.
Let us now consider some different examples.

We start considering $f(q,\hat N)=q^{-4\,\hat N}$. In this case,
assuming that $q\neq 1$, we obtain $x_{n}=
        q^{n-1}\,\frac{1-(1/q^5)^n}{1-1/q^5}$. Again, if $q=1$ we get $x_n=n$. This is a natural
        extension of the result in (\ref{43}) and the operator $\hat N$ satisfies the equality $H_o=q^{{\hat N
}-1}\,\frac{1-(1/q^5)^{{\hat N}}}{1-1/q^5}$.

Let us now take $f(q,\hat N)=e^{(q^2-1)\hat N}$. It is clear that
$\lim_{q\rightarrow\pm1}\,f(q,N)=1$. Applying formula (\ref{45})
we get, once again $x_n=n$ if $q=1$, while, if $q\neq 1$, we
obtain $$x_{n+1}=\frac{q^{n+1}-e^{(n+1)(q^2-1)}}{q-e^{q^2-1}},$$
so that $\hat N$ must satisfies $H_o=\frac{q^{{\hat N }}-e^{{\hat
N} (q^2-1)}}{q-e^{q^2-1}}$. Other examples can be easily
constructed, but we will not do it here.

Summarizing, the so called quons appear just as particular cases
of a much more general strategy, which consists in replacing the
eigenvalues of the harmonic oscillator $\{n\}$ with a general
non-negative sequence $\{x_n\}$ and the o.n. basis
$\Phi_n=\frac{(a^\dagger)^n}{\sqrt{n!}}\Phi_0$ with an arbitrary
o.n. basis of the Hilbert space. This is strongly related to what
has been done in the rest of this paper. A deeper analysis of
these aspects of the theory is in progress.


\section*{Acknowledgements}

This work has been financially supported in part by M.U.R.S.T.,
within the  project {\em Problemi Matematici Non Lineari di
Propagazione e Stabilit\`a nei Modelli del Continuo}, coordinated
by Prof. T. Ruggeri.

\end{document}